\documentclass[10pt,conference,letter paper]{IEEEtran}

\usepackage[T1]{fontenc}
\usepackage{standalone}
\usepackage{amsfonts,amsmath,amssymb,amsthm}

\usepackage[usenames,dvipsnames]{color}
\usepackage{color}

\usepackage{paralist}
\usepackage{relsize}

\usepackage{epsfig}
\usepackage{wrapfig}
\usepackage{graphicx}
\usepackage{subcaption}

\usepackage{flushend}
\usepackage{hyperref}
\usepackage{url}

\usepackage{newclude}

\usepackage[noadjust]{cite} 
\usepackage{xfrac}

\usepackage{pgf}
\usepackage{tikz}
\usetikzlibrary{positioning}
\usetikzlibrary{shapes, arrows, decorations.markings}
\usetikzlibrary{calc}
\usetikzlibrary{fit}

\newtheorem{definition}{Definition}
\newtheorem{theorem}{Theorem}
\newtheorem{claim}{Claim}
\newtheorem{lemma}{Lemma}
\newtheorem{corollary}{Corollary}

\newtheorem{remark}{Remark}

\newcommand{\blowup}[2]{{#1\mathopen{}\left[{#2}\right]\mathclose{}}}
\newcommand{\subpacket}[2]{{ {#1}^{(#2)} }}
\newcommand{\mais}[1]{{\text{\normalfont\footnotesize\textsf{MAIS}}\mathopen{}\left({#1}\right)\mathclose{}}}
\newcommand{\minrank}{{\normalfont\textsf{minrk}}}
\newcommand{\rank}{{\normalfont\text{rank}}}
\newcommand{\ind}[1]{{\mathcal{I}\mathopen{}\left({#1}\right)\mathclose{}}}
\newcommand{\gnscut}[1]{{\text{\normalfont\footnotesize\textsf{GNSCUT}}\mathopen{}({#1})\mathclose{}}}

\newcommand{\mincut}[1]{{\text{\normalfont\footnotesize\textsf{MINCUT}}\mathopen{}\left({#1}\right)\mathclose{}}}

\newcommand{\betavl}{\beta_{\mathsmaller{\mathsmaller{\textnormal{VL}}}}}
\newcommand{\rcp}{r_{\mathsmaller{\mathsmaller{\textnormal{CP}}}}}
\newcommand{\ratecor}{R^{\mathsmaller{\mathsmaller{\textnormal{CO}}}}}
\newcommand{\rateind}{R^{\mathsmaller{\mathsmaller{\textnormal{MU}}}}}

\begin{document}
\title{On Approximating the Sum-Rate for Multiple-Unicasts}
\author{
\IEEEauthorblockN{Karthikeyan Shanmugam, Megasthenis Asteris and Alexandros G. Dimakis} \\
\IEEEauthorblockA{ Department of Electrical and Computer Engineering \\
University of Texas at Austin \\
\texttt{\{karthish, megas\}@utexas.edu, dimakis@austin.utexas.edu}}
}

\date{\today}

\maketitle

\begin{abstract}

We study upper bounds on the sum-rate of multiple-unicasts.
We approximate the Generalized Network Sharing Bound (GNS cut) of the multiple-unicasts network coding problem with~$k$ independent sources.
Our approximation algorithm runs in polynomial time and yields an upper bound on the joint source entropy rate, which is within an~$O(\log^2 k)$ factor from the GNS cut.
It further yields a vector-linear network code that achieves joint source entropy rate within an~$O(\log^2 k)$ factor from the GNS cut, but \emph{not} with independent sources: the code induces a correlation pattern among the sources. 

Our second contribution is establishing a separation result for vector-linear network codes:
for any given field~$\mathbb{F}$ there exist networks for which the optimum sum-rate supported by vector-linear codes over~$\mathbb{F}$ for independent sources can be multiplicatively separated by a factor of~$k^{1-\delta}$, for any constant ${\delta>0}$, from the optimum joint entropy rate supported by a code that allows correlation between sources. 
Finally, we establish a similar separation result for the asymmetric optimum vector-linear sum-rates achieved over two distinct fields~$\mathbb{F}_{p}$ and $\mathbb{F}_{q}$ for independent sources, revealing that the choice of field can heavily impact the performance of a linear network code. 
\end{abstract}

\section{Introduction}

The \emph{multiple-unicasts network coding problem} is one of the fundamental problems in network information theory.
In this problem, $k$ source nodes need to communicate independent information to~$k$ corresponding destinations
through a directed acyclic network.
Information is encoded at the sources and flows through links with limited (typically integral) capacity,
while intermediate nodes create (possibly non-linear) combinations of the incoming messages. 
The canonical question is:
\emph{what is the set of transmission rates supported by a given network~$\mathcal{G}$ with~$k$ independent sources?}
A related objective is determining
\emph{the optimum achievable sum-rate}, \textit{i.e.}, the optimum joint source entropy rate for the~$k$ independent sources.
The problem has been extensively studied (see, \textit{e.g.},~\cite{ho2008network, ramakrishnan2010network, meng2012feasibility} and references therein). It is known that \emph{non-linear} codes are required to achieve the capacity~\cite{dougherty2007}, but few papers have studied the question of approximating the rate for multiple unicasts (e.g.~\cite{langberg2011hardness,Lovett}).

A significant body of work has focused on developing upper bounds 
on the joint source entropy rate for multiple-unicasts with independent sources.
Several of these bounds belong to the class of \emph{edge cut} bounds,
in which the sum-rate is upper bounded by the cumulative capacity of a appropriately selected set of network links. 
\emph{Cut set} bounds are a prominent representative of this family,
but they are outperformed by a newer member of this class: the \emph{GNS (Generalized Network Sharing) cut bound}~\cite{kamath2013study}.
There are several other related bounds including the PdE~\cite{kramer2006edge}, Information dominance~\cite{harvey2006capacity} and Functional dependence~\cite{thakor2009network}. 
With few exceptions (GNS cut and Functional dependence bounds are equivalent), it is not known how these bounds compare.
However, all these bounds share one thing in common with the GNS cut: \emph{they are hard to compute}.
In this paper we shed new light on the GNS cut bound and its approximation power for the multiple-unicasts network coding problem. 

\vspace{\baselineskip}
\noindent\textbf{Our Contributions:}
\begin{enumerate}
\item We tensorize the GNS-cut bound as follows: We use an argument that uses strong graph products to obtain a sequence of rate upper bounds that are valid for vector-linear codes -- we show that the weakest bound in this sequence is the GNS cut bound. 
   
\item We define a new communication problem that we call the \emph{relaxed-correlated} multiple-unicasts.
In this problem independence across sources is relaxed: the code designer is allowed to introduce \emph{any} correlation structure in the sources in order to maximize the joint source entropy rate.
 GNS cut is an upper bound on the optimum joint source entropy rate for this relaxed-correlated multiple-unicasts problem. 

\item We develop a polynomial time algorithm to provably approximate the GNS cut bound from above within an~$O(\log^2 k)$ factor, where $k$ is the number of sources in the network. 
Our algorithm also yields a \emph{vector-linear code} for the
relaxed-correlated sources problem achieving joint source entropy rate within an~$O(\log^2 k)$ factor from the optimum over all (even non-linear) network codes. 

\item One important question is how the finite field used by the vector-linear code influences the sum-rate. We show that the choice of the field matters tremendously. 
For \emph{any two fields} $\mathbb{F}_p$ and $\mathbb{F}_q$ and for any $\delta>0$, there exist multiple-unicast networks for large $k$ such that the optimal sum-rates over $\mathbb{F}_p$ and $\mathbb{F}_q$ differ by a factor of $k^{1-\delta}$, for (Theorem~\ref{thm:sep}). 
Note that a $\sfrac{1}{k}$-approximation can be achieved by having a single unicast and ignoring all other sources. 
Our result shows that this kind of separation can almost be caused by a poor choice of field. 
This partially negatively answers an open problem stated recently in~\cite{Lovett}, asking whether vector-linear codes can approximate the network capacity within a logarithmic factor. 
Our result shows that the answer is negative for the sum-rate over a fixed field. 
This relies on a similar result for the symmetric-rates~ (\cite{blasiak2013graph,lubetzky2009nonlinear}).

We also show that any field can be bad for some network: for \emph{any given field} $\mathbb{F}$, 
there exists a multiple-unicasts network (with sufficiently large~$k$) for which the optimum vector-linear joint entropy rates for independent and correlated sources are separated by a factor of $k^{1-\delta}$, for any constant~${\delta >0}$ (Theorem~\ref{thm:sep}).
\end{enumerate}
Note that our results do not rule out the approximation of the optimum sum-rate for multiple-unicasts by linear codes in general. They do imply, however, that the achievability must use a field that depends on the network.
It is possible that the optimal vector-linear code sum-rate (over the best field for that network) is close to the GNS cut bound and the optimum non-linear sum-rate.
Note that previous results imply that this is not possible for the symmetric-rate~\cite{blasiak2013graph}.

Our developments rely on connections between the relaxed-correlated multiple-unicasts network problem and the \emph{index coding problem} established in~\cite{shanmugam2014bounding}.
In the index coding problem, a single broadcasting agent needs to communicate $n$ distinct messages to $n$ receivers (one message per receiver) over a noiseless broadcast channel. 
A subset of the source messages is available as side-information to each receiver. 
The objective is to design a broadcast scheme that uses minimum number of transmissions to deliver the~$n$ messages.
This has been well studied~\cite{bar2011index, birk1998informed, maleki2012index, blasiak2013broadcasting}.
Computing the minimum number of necessary transmissions, even in the case of scalar linear coding schemes, is NP-hard and hard to even approximate within a constant factor~\cite{peeters1996orthogonal, bar2011index}. 
The multiple-unicasts network coding problem can be reduced to an index coding problem when it comes to exact solvability of specific rate tuples~\cite{effros2012equivalence}. 
However, under the connection of~\cite{effros2012equivalence} it is not clear if an approximately 'good' symmetric solution for the reduced index coding problem can be converted to a 'good' network coding solution.

\section{Definitions}
We begin with a set of formal definitions that are useful for our subsequent developments. 
\begin{definition}
\emph{\textbf{(Directed Index Coding)}}
Consider a set of $n$ independent messages (\emph{symbols}) $\mathbf{x}_i \in \mathbb{F}^{p}$, $i = 1, \hdots, n$, each consisting of $p \in \mathbb{N}_{+}$ packets (\emph{subsymbols}) in some alphabet~$\mathbb{F}$, and a set of $n$ users $\lbrace 1, \hdots, n \rbrace$, such that user $i$:
   \begin{inparaenum}
    \item wants message $x_i$, and 
    \item has messages $x_j$, $j \in S_i \subseteq \{1,\hdots, n \}\backslash \lbrace i \rbrace$ as side-information.
    \end{inparaenum}
A sender wishes to broadcast all $n$ messages to the corresponding users over a noiseless channel.
The objective is to design a coding scheme that minimizes the number of transmissions required for all users to decode their respective messages.
$\hfill \lozenge$
\end{definition}


An Index Coding instance is fully characterized by its \textit{side-information graph} $G$.
The side-information graph $G$ is a directed graph on~$n$ vertices corresponding to the~$n$ users. 
An edge $(i,j)$ exists in $G$ if and only if~${j \in S_i}$, \textit{i.e.}, user~$i$ \emph{has} message~$\mathbf{x}_j$ as side-information.


Let 
${\mathbf{x}= \left[\mathbf{x}_1^{\mathsmaller{T}} ~\mathbf{x}_2^{\mathsmaller{T}} \cdots \mathbf{x}_n^{\mathsmaller{T}} \right]^{\mathsmaller{T}}}$ 
be the {$(pn)$-dimensional} vector formed by stacking the~$n$ symbols ${\mathbf{x}_1, \hdots, \mathbf{x}_n \in \mathbb{F}^{p}}$.
The sender transmits one symbol (or equivalently $p$-subsymbols)  per \emph{channel use}. An $\left( \mathbb{F},p,n,r \right)$ \textit{vector-linear index code} for this problem consists of $r$ linear combinations of symbols in $\mathbf{x}$ over a field $\mathbb{F}$ that satisfies the decodability criterion at every user.

\color{black}
         
The \textit{broadcast rate} $\betavl^{\mathbb{F}}\mathopen{}\left(G,\mathcal{C}\right)$ of an $\left( \mathbb{F},p,n,r \right)$-vector-linear index code $\mathcal{C}$ is the ratio $\sfrac{r}{p}$;
\emph{the number of channel uses} required for all users to receive their message.%
\footnote{Recall that a channel use is the transmission of a symbol, or equivalently the transmission of $p$-subsymbols.}



\color{black}

\begin{definition}
   \label{def:mult-unicast-network}
   (\emph{\textbf{Multiple-Unicasts (MU) Network}})
   A \textit{multiple-unicasts network} instance is an acyclic directed network $\mathcal{G}(\mathcal{N}, \mathcal{E})$ 
   on a set $\mathcal{N}$ of nodes, with the following components:
   \begin{enumerate}
   \item $\mathcal{E}$ is the set of links (edges) in the network.
   Links have unit capacity; they carry at most one bit per channel use. 
   We use $c_{(a,b)}$ to denote the total capacity from node~$a$ to node~$b$, \textit{i.e.}, the number of links from $a$ to $b$.
   Finally, $h(e)$ and $t(e)$ denote the head and tail of an edge ${e \in \mathcal{E}}$, respectively.
   \item  (\emph{Source/Destination nodes}) $\mathcal{S}\triangleq\lbrace s_1,s_2, \hdots, s_k\rbrace \subseteq {\cal N}$ is a set of $k$ source nodes, and $\mathcal{T}\triangleq\lbrace t_1,t_2, \hdots, t_k \rbrace \subseteq {\cal N}$ is a set of $k$ destination nodes.
   \item (\emph{Source links})
   ${ \mathcal{E}_i \subset \mathcal{E}}$ is a set of $\mathrm{mincut}(s_i,t_i)$ edges with \emph{no} tail and head~$h(e)=s_i, \; \forall e \in {\cal E}_i$, $i=1,\hdots, k$.
   Here, $\mathrm{mincut}(s_i,t_i)$ is the number of unit-capacity links in the minimum cut between source $s_i$ and destination $t_i$.
      We refer to ${\cal E}_i$ as the set of \emph{source links} of source $s_i$.
   \end{enumerate}
   Each source node $s_i$ wants to transmit information to its corresponding destination $t_i$, ${i=1,\hdots, k}$.
   Information is fed into the network through the source links  $\cup_{i=1}^{k}\mathcal{E}_i$.
   $\hfill \lozenge$
\end{definition}      
The \emph{multiple-unicasts network coding problem} is the problem of designing a \emph{network code}: the set of rules that govern how information is encoded and flows through the network.
One of the canonical objectives of \emph{multiple-unicasts network coding} is to 
maximize the total amount of information transmitted through the network per channel use,
\textit{i.e.}, to maximize the \emph{joint source entropy rate}.
Here, we focus only on \emph{vector-linear codes}, \textit{i.e.},
codes in which encoding and decoding involve only vector-linear operations. 
\begin{definition}
   \label{def:mult-unicast-net-code}
   \emph{(\textbf{Vector Linear MU Network Code})}
   An $(\mathbb{F},p,m,r)$ vector-linear MU network code ${\cal C}$ is a collection of vectors ${\mathbf{z}_e \in \mathbb{F}^p}, \forall {e \in \mathcal{E}}$ that depend on the aggregate source message vector $\mathbf{x} \in \mathbb{F}^r$ (consisting of $r$ independent subsymbols)  satisfying:
    \begin{enumerate}
    \item \emph{Coding at intermediate nodes:} For \emph{source link} $e$, $\mathbf{z}_e$ is a linear combination of sub-symbols in $\mathbf{x}$. For each \emph{non-source} link $e \in {\cal E}$, $\mathbf{z}_e$ is a linear combination of $\mathbf{z}_a$'s of the edges incident on it, i.e. $\{ \mathbf{z}_a \}_{a:h(a)=t(e)}$. 
\item \emph{Decoding at destinations:} At every destination $t_i$, every variable $\mathbf{z}_e$  for $e \in {\cal E}_i$, is linearly decodable from information flowing into $t_i$, i.e. $\{ \mathbf{z}_a \}_{a:t(a)=t_i}$.
 \item \label{req:independent-sources} \emph{Independence between sources:} The variables of one source, i.e. $\{\mathbf{z}_e\}_{e \in {\cal E}_i}$ are mutually independent of those of other sources.
    \end{enumerate}
  \end{definition}  
The joint source entropy rate achieved by such a code is equal to~$\sfrac{r}{p}$ bits per channel use. 
Due to the independence among sources, the joint source entropy rate is equal to the \emph{sum-rate} of the~$k$ sources.
We use $\rateind\mathopen{}\left({\mathcal{G}}; \mathbb{F}\right)$ to denote the 
\emph{optimum sum-rate achievable over all vector-linear network codes defined over the field $\mathbb{F}$},
and $\rateind\mathopen{}\left({\mathcal{G}}\right)$ to denote the optimum vector-linear sum-rate over all fields.

\textbf{Relaxed-Correlated Sources.}
For our developments it is useful to consider a variant of the multiple-unicasts network coding problem, in which the requirement that source information is independent across sources is overlooked. 
We refer to the modified version as the problem of \emph{relaxed-correlated sources}.
In the modified problem, we still seek to maximize the maximum joint source entropy, but allow arbitrary correlations among sources.
\begin{definition}
   \label{def:mult-unicast-net-code-corr}
   \emph{(\textbf{Vector-Linear Relaxed-Correlated MU Network Code})}
   A vector-linear code~$\mathcal{C}$ for the multiple-unicasts network coding problem with \emph{relaxed-correlated sources} is defined as in Def.~\ref{def:mult-unicast-net-code} omitting requirement~\eqref{req:independent-sources}. 
   $\hfill \lozenge$	          
\end{definition}

We use ${\ratecor\mathopen{}\left({\mathcal{G}}; \mathbb{F}\right) = r/p}$ to denote the optimum joint source entropy rate achievable by vector-linear codes over a given field $\mathbb{F}$ in the \emph{relaxed-correlated sources} problem, and $\ratecor\mathopen{}\left({\mathcal{G}}\right)$ to denote the optimum rate over all fields, accordingly. 
Clearly, $\rateind\mathopen{}\left({\mathcal{G}}\right) \le \ratecor\mathopen{} \left({\mathcal{G}}\right)$.

\begin{remark}
We emphasize that the optimal joint source entropy rate $\ratecor \left({\mathcal{G}}\right)$ 
in the relaxed-correlated sources problem is achieved for some (unspecified) source correlation pattern.
The correlation pattern is an additional degree of freedom towards maximizing the joint source entropy rate; it is \emph{not} a code design constraint. 
\end{remark}

\begin{remark}
   In Def.~\ref{def:mult-unicast-network}, we require $\lvert\mathcal{E}_i\rvert=\mincut{s_i,t_i}$.
   This is only a useful convention and does not affect the value of $\rateind(\mathcal{G})$.
   It does, however, affect $\ratecor(\mathcal{G})$.
   In this work, we upper bound~$\rateind(\mathcal{G})$ by developing bounds on~$\ratecor(\mathcal{G})$.
   Hence, the convention becomes essential. 
   \end{remark}


\color{black}

\section{Bounds on the Sum-Rate of an MU Network}  \label{sec:bounduni}            

\newcommand{\Ecard}{m}
\newcommand{\Lcard}{\widehat{m}}
We develop  upper bounds on~$\rateind(\mathcal{G})$, the optimum sum-rate supported by an MU network with independent sources using vector-linear codes. 
In fact, our bounds are developed for~$\ratecor(\mathcal{G})$, 
the optimum vector-linear joint source entropy rate in the relaxed-correlated sources problem.
Our results extend those in~\cite{shanmugam2014bounding}.

\subsection{From Multiple-Unicasts Network Coding to Index Coding}
\label{G-of-calG}
Consider a multiple-unicasts network $\mathcal{G}$ with $k$ sources and $\Ecard$ links.
Let $\mathcal{G}^\prime$ be a directed cyclic network constructed from $\mathcal{G}$ 
by setting 
${t(e) = t_i}$, ${\forall e \in \mathcal{E}_i}$, ${i=1, \hdots, k}$,
\textit{i.e.}, setting the destination node $t_i$ to be the tail of every source link of source~$s_i$.

Let $G$ be the \textit{(reversed)
\footnote{We refer to $G$ as the \emph{reversed} line graph of~$\mathcal{G}$ because the direction of its edges is reversed compared to the typical definition of a line graph.}
line graph} of ${\mathcal{G}}^\prime$,
\textit{i.e.}, a directed graph on $\Ecard$ vertices corresponding to the $m$ links in ${\mathcal{G}}^\prime$,
with a directed edge from vertex $v$ to $\widehat{v}$ corresponding to links $e$ and $\widehat{e}$, respectively, iff  $h(e)=t(\widehat{e})$ in $\mathcal{G}^\prime$.


\color{black}
\begin{theorem}[\cite{shanmugam2014bounding}] 
	\label{thm:duality}
	Consider a multiple-unicasts network~$\mathcal{G}$ with~$\Ecard$ links,
	and a vector-linear code~$\mathcal{C}$ with correlated sources, achieving joint source entropy rate~$r$.
	The dual code $\mathcal{C}^{\perp}$ is a vector-linear index code achieving rate $m-r$  in the index coding instance with side-information graph $G$ constructed based on~$\mathcal{G}$ as described in Section~\ref{G-of-calG}.
\end{theorem}
\begin{corollary}
	\label{cor:duality}
	 If $G$ is the directed graph constructed based on the network~$\mathcal{G}$ as described in Section~\ref{G-of-calG}, then
	\begin{align}   
		\ratecor\mathopen{}\left( \mathcal{G} \right)
		\;=\; 
		\Ecard - \betavl\mathopen{}\left( G \right).
		\nonumber
	\end{align}
\end{corollary}


We exploit the connection established in Cor.~\ref{cor:duality}
to develop upper bounds on the joint source entropy rate $\ratecor\mathopen{}\left( \mathcal{G} \right)$,
through properties of the side information graph~$G$ of the associated index coding problem.
\begin{definition}
   \label{def:mais}
   $\mais{G}$ of a directed graph $G$ is the cardinality of the largest set ${\widehat{V} \subseteq V(G)}$
   such that the subgraph of $G$ induced by~$\widehat{V}$ is acyclic.
\end{definition}
It is known that the size of the maximum acyclic subgraph of $G$ is a lower bound on $\betavl\mathopen{}\left( G \right)$. Tighter bounds can be obtained via graph tensorization.
\begin{lemma}
	\label{lem:beta_VL_gt_mais_q}
	The optimum broadcast rate ${\betavl(G)}$ of an index coding instance with side-information graph $G$, satisfies
	\begin{align} 
		\sqrt[q]{\mais{\otimes^{q}G}}
		\;\le\;
		\betavl\mathopen{}\left( G \right),
		\quad \forall q\in \mathbb{Z}^{+},\nonumber
	\end{align} 
	where $\otimes^q$ denotes the strong product of $G$ with itself $q$ times.
\end{lemma}
\begin{IEEEproof}
 See proof of Theorem~\ref{thm:shannoncapbound} in the Appendix.
\end{IEEEproof}

\begin{theorem}
\label{thm:uni}
   Consider a multiple-unicasts network $\mathcal{G}$ with~$k$ sources and~$\Ecard$ links.
   Further, let $G$ be the digraph on $\Ecard$ vertices obtained from $\mathcal{G}$ as described in Section~\ref{G-of-calG}.
   Then,
   \begin{align}   
      \rateind\mathopen{}\left( \mathcal{G} \right)
      \,
      \le
      \,
      \ratecor\mathopen{}\left( \mathcal{G} \right)
      &=
      \Ecard - \betavl\mathopen{}\left( G \right) \nonumber \\
      &\le
      \;
      \Ecard - \sqrt[q]{\mais{\otimes^{q}G}},
      \quad
      q\in \mathbb{Z}^{+}.\nonumber
   \end{align}
\end{theorem}
\begin{IEEEproof}
The proof follows from Thm.~\ref{thm:duality} and Lem.~\ref{lem:beta_VL_gt_mais_q}.
\end{IEEEproof}
\color{black}

\subsection{Comparison with GNS cut bound}

We compare the bounds of Theorem~\ref{thm:uni}
with the GNS cut bound;
the weakest among the former (\textit{i.e.}, for ${q=1}$) is at least as good as the latter.
A more careful application of the GNS cut approach (on a slightly modified network) 
reveals that the two bounds are effectively equal.
First, recall the definition of the GNS cut:


\begin{definition}[\cite{kamath2013study}]
	\label{definition-gnscut}
   A GNS cut of a multiple-unicasts network~${\mathcal{G}(\mathcal{V}, \mathcal{E})}$ with $k$ sources, is a subset~${S \subset {\cal E}}$ such that 
   for $\mathcal{G}-S$ (\textit{i.e.}, the network obtained by removing the links in $S$ from $\mathcal{G}$) the following holds:
   there exists a permutation ${\pi:[k]\rightarrow [k]}$ such that
   ${\forall i,j \in [k]}$, if $\pi(i) \geq \pi(j)$, then no path exists from source $s_i$ to destination $t_j$.
\end{definition}

The size of the smallest (in terms of capacity) GNS cut, denoted by $\gnscut{\mathcal{G}}$,
is an upper bound on the \emph{non-linear} sum-rate of the multiple-unicasts problem with independent sources~\cite{kamath2013study}.
\begin{theorem}
   \label{thm:m-minus-mais-lt-gnscut}
   Consider a multiple-unicasts network~$\mathcal{G}$ with~$k$ sources and~$\Ecard$ links.
   Let~$G$ be a digraph on $\Ecard$ vertices constructed based on~$\mathcal{G}$ as described in Section~\ref{G-of-calG}.
   Then,
   \begin{align}
      \Ecard-\mais{G} \;\leq\; \gnscut{\mathcal{G}}. \nonumber
   \end{align}
\end{theorem}
\begin{IEEEproof}
   See Appendix, Section~\ref{sec:proof:m-minus-mais-lt-gnscut}.
\end{IEEEproof}
\noindent In other words, the bounds of Theorem~\ref{thm:uni} are at least as tight as the GNS cut bound. 

The GNS cut technique can be slightly strengthened to yield an upper bound exactly equal to $m-\mais{G}$.
We achieve that by obtaining the GNS cut bound on a modified, yet equivalent network.
Given a multiple-unicasts network~$\mathcal{G}(\mathcal{V}, \mathcal{E})$ with~$k$ sources and~$\Ecard$ links,
consider a network~$\widetilde{{\mathcal{G}}}(\widetilde{\mathcal{V}}, \widetilde{\mathcal{E}})$ obtained from $\mathcal{G}$ as follows:%

\begin{enumerate}
	\item Introduce $k$ nodes ${\widetilde{s}_1, \hdots, \widetilde{s}_k}$ to $\mathcal{G}$, \textit{i.e.}, ${\widetilde{\mathcal{V}} = \mathcal{V} \cup \lbrace \widetilde{s}_i\rbrace}_{i=1}^{k}$.
	\item Set $t(e)=\widetilde{s}_i,\; \forall e \in \mathcal{E}_i$, $i =1,\hdots, k$, that is, set $\widetilde{s}_i$ as the tail of all source links of source $s_i$.
	\item Introduce a set $\widetilde{\mathcal{E}}_i$ of  $|\mathcal{E}_i|$ new links with head $\widetilde{s}_i$ and no tail, for all ${i \in \lbrace 1, \hdots, k\rbrace}$.
\end{enumerate}
The modified network $\widetilde{{\mathcal{G}}}$ is a multiple-unicasts network with~$k$ sources $\widetilde{s}_1, \hdots, \widetilde{s}_k$ and respective destinations $t_1, \hdots, t_k$.
One can verify that ${\ratecor(\widetilde{\mathcal{G}})=\ratecor(\mathcal{G})}$.
The key difference is that the $|\mathcal{E}_i|$ source links of source $s_{i}$ in $\mathcal{G}$ have become regular links in~$\widetilde{\mathcal{G}}$ and can be used in a GNS cut.
Thus, the bound obtained on the modified network is potentially tighter, \textit{i.e.}, $\gnscut{\widetilde{\mathcal{G}}} \le \gnscut{{\mathcal{G}}}$.

\begin{theorem}
\label{thm:m-minus-mais-eq-gnscut-Gtilde}
Consider a multiple-unicasts network $\mathcal{G}$ with~$k$ sources and $\Ecard$ links.
Let $G$ be the digraph on $\Ecard$ vertices obtained from $\mathcal{G}$ as described in Section~\ref{G-of-calG}, and $\widetilde{{\mathcal{G}}}$ the modified network constructed as described above.
Then, any feasible feedback vertex set of~$G$ corresponds to a GNS cut in~$\widetilde{{\cal G}}$ with the same capacity.
In turn,
   \begin{align}
      \Ecard-\mais{G} = \gnscut{\widetilde{\mathcal{G}}}. \nonumber
   \end{align}
\end{theorem}
\begin{IEEEproof}
The proof relies on showing that each GNS cut in $\widetilde{\mathcal{G}}$ corresponds to a \emph{Feedback Vertex Set} (FVS) of equal size in the digraph $G$.
A FVS $\mathcal{F}$ is a set of vertices 
such that the subgraph of $G$ induced by ${V(G)-\mathcal{F}}$ is acyclic.
By definition, 
$\Ecard-\mais{G}$ is the cardinality of the \emph{minimum feedback vertex set} in $G$.

See Appendix, Section~\ref{sec:proof::m-minus-mais-eq-gnscut-Gtilde} for the complete proof.   
\end{IEEEproof}


\begin{remark}
   $\gnscut{\mathcal{G}}$ is an upper bound on the \emph{non-linear} sum-rate of the multiple-unicasts network coding problem on $\mathcal{G}$ with independent sources (Thm.~$1$ in Chapter $2$ of~\cite{kamath2013study}, Thm.~$2$ in~\cite{kamath2011generalized}).
   Those results can be generalized to show that $\gnscut{\mathcal{G}}$ also upper bounds the non-linear joint source entropy rate in the problem of relaxed-correlated sources.
   By Theorem~\ref{thm:m-minus-mais-eq-gnscut-Gtilde}, it follows that ${\Ecard-\mais{G}}$ is also an upper bound on the optimum \emph{non-linear} joint source entropy rate in both problems.
\end{remark}

%
%

\section{Approximating the GNS cut  bound}
Determining the GNS cut bound for a given network is computationally hard problem in general \cite{kamath2013study}.
We describe an algorithm to approximately compute the GNS cut bound for a given acyclic network $\tilde{{\cal G}}$. 
We exploit the connection of Theorem~\ref{thm:m-minus-mais-eq-gnscut-Gtilde}
and the special structure of a multiple-unicasts network~$\widetilde{\mathcal{G}}$,
and we utilize known approximation algorithms for the \emph{Feedback Vertex Set problem} on a diagraph.

The Feedback Vertex Set (FVS) problem,
\textit{i.e.}, the problem of finding the smallest FVS in a given digraph~$G$, 
 is NP-complete \cite{seymour1995packing}.
The LP dual of its LP relaxation is the \emph{fractional cycle packing problem}~\cite{even1998, seymour1995packing}.
A fractional cycle packing is a function $q(C)$ from the set of cycles~$\mathcal{C}$ in~$G$ to~${[0,\,1]}$, satisfying
$\sum_{ C \in \mathcal{C}:v \bigcap C \neq \emptyset} q(C) \leq 1$, for each~${v \in V(G)}$.
Letting $|q| =
\sum_{C\in \mathcal{C}} q(C)$, the \emph{fractional cycle packing number}~$\rcp(G)$ of~$G$ is defined to be the maximum of~${|q|}$ taken over all
fractional cycle packings~$q$ in~$G$.
Clearly,~${\rcp(G) \le \lvert \mathcal{F} \rvert}$ for \emph{all} feedback vertex sets~$\mathcal{F}$ in~$G$.
By definition, 
$\Ecard-\mais{G}$ is the cardinality of the \emph{minimum feedback vertex set} in $G$.
Therefore,
\begin{align}
   \rcp(G) 
   \;\le\;
   \Ecard - \mais{G}.
   \nonumber
\end{align}

An \textit{optimal} fractional cycle packing \cite{nutov2004packing} \cite{even1998} (or an $(1+\epsilon)$ approximation, $\epsilon>0$) can be computed in polynomial time (in $m$, $\epsilon^{-1}$). 
A feasible fractional cycle packing on~$G$ can be suitably rounded to yield a FVS~$\mathcal{F}$ with cardinality~$|\mathcal{F}| \le {\rcp(G) \cdot O(\log{\Ecard}\log\log{\Ecard})}$. 
We conclude that for any directed graph~$G$ on~$m$ vertices,
\begin{align}
   \Ecard-\mais{G}
   \;
   \le
   \;
   \rcp(G) \cdot O( \log{\Ecard}\log\log{\Ecard}).
   \label{rcp-sandwich-on-m-mais-loose}
\end{align}

Note that~\eqref{rcp-sandwich-on-m-mais-loose} holds for arbitrary digraphs.
But $G$ has special structure as it is the (reverse) line-graph of a multiple-unicasts network~$\mathcal{G}^\prime $ (itself a modification of a network~$\mathcal{G}$) as described in Section~\ref{G-of-calG}.
Any feedback \emph{vertex} set in $G$ maps trivially to a feedback \emph{edge} set in~$\mathcal{G}^\prime $.
All cycles in $\mathcal{G}^\prime $ go through the~$k$ source nodes of~$\mathcal{G}$. 
There exist polynomial-time algorithms that exploit this additional structure to compute a feedback edge set (in turn, a feedback vertex set in~$G$) with cardinality within a~$O(\log^2{k})$ factor from~$\rcp(G)$~\cite{even1998}.
\begin{theorem}\label{thm:app}
Consider a multiple-unicasts network $\mathcal{G}$ with $k$ sources and $\Ecard$ unit-capacity links.
Let $G$ be the digraph on $\Ecard$ vertices obtained from $\mathcal{G}$ as described in Section~\ref{G-of-calG}.
Then, 
\begin{align}
	\rcp(G)
	\; & \le \;
	\Ecard-\betavl(G) \nonumber\\
	\; &\le \;
	\Ecard-\mais{G} 
	\; \le \;
	{\rcp(G) \cdot O\bigl({\log^2{k}}\bigr)}, \nonumber
\end{align}
where $\rcp(G)$ is the \emph{fractional cycle packing number} of~$G$.
Further, $\rcp(G)$ also equals the joint source entropy rate supported by a feasible (and polynomial-time computable) vector-linear multiple-unicasts network code for the relaxed-correlated sources problem on~$\mathcal{G}$.
\end{theorem}  
\begin{IEEEproof}
See Appendix, Section~\ref{sec:proof:app}.
\end{IEEEproof}

\section{Price of Independence}
The GNS cut, similar to the novel bounds of Theorem~\ref{thm:uni},
upper bound the optimum vector-linear joint source entropy rate for the relaxed-correlated sources, and in turn for independent sources since ${\ratecor ({\mathcal{G}}) \ge \rateind ({\mathcal{G}})}$.
However, it remains unclear how the gap between the two rates scales.
The following Theorem takes a step towards addressing this question.
\begin{theorem}\label{thm:sep}
   For any prime field $\mathbb{F}_p$, for any constant ${\delta>0}$, there is a $k$ sufficiently large and there exists a family of multiple-unicasts network instances ${\mathcal{G}}$ with~$k$ sources ($k$ sufficiently large) for which 
   $\ratecor \left( {\mathcal{G}};\mathbb{F}_p \right) \ge k^{1- \delta} \cdot \rateind\mathopen{}\left( {\mathcal{G}}; \mathbb{F}_p \right)$. 
   Further, for any two fields $\mathbb{F}_p$ and $\mathbb{F}_q$, for any $\delta>0$, there is a large enough $k$ and a multiple-unicasts network $\mathcal{G}$ such that $\rateind\mathopen{}\left(\mathcal{G}; \mathbb{F}_q \right) \geq k^{1-\delta}\cdot \rateind\mathopen{}\left( \mathcal{G}; \mathbb{F}_p  \right)$, .
\end{theorem}   
\begin{IEEEproof}
 See Appendix, Section~\ref{sec:proof:sep}.
\end{IEEEproof}
Theorem~\ref{thm:sep} effectively states that for a \emph{fixed} field,
there exists networks for which the optimum sum-rate over all vector-linear codes \emph{over that field}
is almost a $k$-factor away from the GNS cut bound.
Second, when designing a vector-linear code for a given multiple-unicasts network, the choice of field can have a tremendous impact on performance: it can affect the achievable sum-rate by almost a factor of $k$.

\section{Conclusions}
We presented a sequence of upper bounds on the sum-rate for multiple-unicasts that are valid for vector-linear codes.
The first bound in this sequence is equivalent to the GNS cut bound, which also holds for non-linear codes. 
Further, we showed that the GNS cut bound can be approximated within an~ $O(\log^2{k})$ factor in polynomial time. 
This is, to the best of our knowledge, the only case in the family of generalized cut-set bounds~\cite{kamath2013study} that can be efficiently approximated.
Finally, we show the importance of the field used by the vector-linear code: the GNS cut and the capacity can be very far from the best vector-linear code over a poorly chosen field.

\bibliographystyle{IEEEtran}
\bibliography{GLRCbibv2}

\appendix

\subsection{Proof of Theorem \ref{thm:m-minus-mais-lt-gnscut}}
\label{sec:proof:m-minus-mais-lt-gnscut}
Recall that $\mathcal{G}^\prime$ is the directed network obtained from~$\mathcal{G}$ by setting the destination node $t_i$ to be the tail of each source link of source $s_i$, ${i=1, \hdots, k}$. (Section~\ref{sec:bounduni}).
Further, $G$ is the (reversed) line digraph of $\mathcal{G}^\prime$.
Any set of vertices lying on a cyclic path in $G$ corresponds to a set of edges forming a cycle in $\mathcal{G}^\prime$.
Hence, $\Ecard-\mais{G}$ equals the cardinality of the minimum feedback edge set of the cyclic network $\mathcal{G}^\prime$, \textit{i.e.}, the smallest set of (unit-capacity) edges that need to be removed from $\mathcal{G}^\prime$ to obtain an acyclic network.
To show the desired result, it suffices to show any GNS cut in ${\mathcal{G}}$ is a feedback edge set in $\mathcal{G}^\prime$.

Any cycle in $\mathcal{G}^\prime$ must contain at least one of the edges connecting a destination node $t_i$ to its source node $s_i$:
these are the only links modified to obtain  $\mathcal{G}^\prime$ from $\mathcal{G}$, and the latter is an acyclic network.   
It turn, 
all cycles in $\mathcal{G}^\prime$ are of the form ${t_{i}, s_{i}, \hdots, t_{i}}$.

Let $S$ be a GNS cut in $\mathcal{G}$.
By definition, there exists a permutation $\pi:[k]\rightarrow [k]$ such that
if $\pi(i) \geq \pi(j)$, no path exists from $s_i$ to $t_j$ in ${\mathcal{G}-S}$.
We want to show that ${\mathcal{G}^\prime-S}$ is acyclic.
Assume, for the sake of contradiction, that this is not the case, and let $\mathcal{C} \subseteq \lbrace 1,\hdots, k\rbrace$ be the set of indices such that a source edge from $t_i$ to $s_i$, $~\forall i \in C$ lies on a cycle.
Let $i^{\star} = \max_{i \in \mathcal{C}} \pi(i)$.
Consider a cycle in ${\mathcal{G}^\prime-S}$ going through an edge from $t_{i^{\star}}$ to $s_{i^{\star}}$; it must be of the form $t_{i^{\star}}, s_{i^{\star}}, \hdots, t_{i^{\star}}$. Without loss of generality, only one of the source edges from $t_{i^{\star}}$ to $s_{i^{\star}}$ occurs in this cycle.
Since $S$ is a GNS cut of $\mathcal{G}$, no path exists in $\mathcal{G}$ from $s_{i^{\star}}$ to $t_{i^{\star}}$.
We conclude that a path from $s_{i^{\star}}$ to $t_{i^{\star}}$ in $\mathcal{G}^\prime$ must use edges that are not available in $\mathcal{G}$, that is, edges from $t_j$ to $s_j$ for some $j \in \lbrace 1, \hdots, k\rbrace$.
Let $j^{\star}$ be the source node such that a source edge from $t_{j^{\star}}$ to $s_{j^{\star}}$ is the first source edge appearing in the path from $s_{i^{\star}}$ to $t_{i^{\star}}$.
Then, the path from $s_{i^{\star}}$ to $t_{j^{\star}}$ uses only edges in $\mathcal{G}-S$ (otherwise it would go through another source edge contradicting the fact that edge from $t_{j^{\star}}$ to $s_{j^{\star}}$ is the first source edge in the cycle after $s_{i^{\star}}$). Hence, there is a path in ${\mathcal{G} -S}$ from $s_{i^{\star}}$ to $t_{j^{\star}}$, with $\pi(i^{\star}) > \pi (j^{\star})$ which is a contradiction.

\subsection{Proof of Theorem \ref{thm:m-minus-mais-eq-gnscut-Gtilde}}
\label{sec:proof::m-minus-mais-eq-gnscut-Gtilde}
Recall that $G$ is the (reversed) line graph of $\mathcal{G}^\prime$, 
the cyclic network obtained from $\mathcal{G}$ by connecting each destination node $t_{i}$ to the source links of the source node $s_{i}$, $i =1\hdots, k$, as described in Section~\ref{G-of-calG}.

The quantity $\Ecard-\mais{G}$ is the cardinality of the minimum feedback vertex set in $G$, which in turn equals the cardinality of the minimum feedback edge set (FES) in~$\mathcal{G}^\prime$, \textit{i.e.},
   \begin{align}
      \Ecard-\mais{G}
      \quad = \quad
      \min_{\mathcal{F}^\prime \text{ is a FES in } \mathcal{G}^\prime}|\mathcal{F}^\prime|.
      \label{eq:m-mais-eq-min-FES-Gprime}
   \end{align}
We will show that the right hand side of~\eqref{eq:m-mais-eq-min-FES-Gprime} is equal to $\gnscut{\widetilde{\mathcal{G}}}$, \textit{i.e.},
   the cardinality of the smallest GNS cut in~$\widetilde{\mathcal{G}}$.
   In fact, we will show that for every FES $\mathcal{F}^{\prime}$ in~$\mathcal{G}^\prime$,
   there exists a GNS cut-set  $\widetilde{\mathcal{F}}$ in $\widetilde{\mathcal{G}}$ with $|\widetilde{\mathcal{F}}|=|\mathcal{F}^{\prime}|$, and vice versa.
   
   Ignoring the source links in $\widetilde{\mathcal{G}}$ (links with head vertex $\widetilde{s}_{i}$ for some $i \in [k]$ and no tail vertex),
   we consider the following trivial one-to-one mapping ${\cal M}$ between the links of $\widetilde{\mathcal{G}}$ and those of~$\mathcal{G}^\prime$:
   \begin{itemize}
   	\item Each of the $\mincut{s_i, t_i}$ links from $t_i$ to $s_i$ in $\mathcal{G}^\prime$ is mapped to a link from $\widetilde{s}_i$ to $s_i$ in $\widetilde{\mathcal{G}}$.
	\item All remaining links are common in both networks.
   \end{itemize}
   
   Consider an arbitrary FES $\mathcal{F}^\prime$ in $\mathcal{G}^\prime$.
   Let $\widetilde{\mathcal{F}}$ be the image of $\mathcal{F}^\prime$ under the mapping $\mathcal{M}$.
   We will show that $\widetilde{\mathcal{F}}$ is a GNS cut-set in~$\widetilde{\mathcal{G}}$.
   \begin{claim}
   \label{not-gns-implies-sequence}
   Consider a $k$-unicast network $\mathcal{G}(\mathcal{V}, \mathcal{E})$ and a subset of links $\mathcal{F} \subset \mathcal{E}$.
   If $\mathcal{F}$ is \emph{not} a GNS cut-set in $\mathcal{G}$, then
   there exists a source-destination pair $s_{i}, t_{i}$ with a path from $s_{i}$ to $t_{i}$ in $\mathcal{G}-\mathcal{F}$, 
   or a sequence of $r \ge 2$ distinct indices $i_{1}, \hdots, i_{r} \in \lbrace 1, \hdots, k \rbrace$ such
   source $s_{i_j}$ has a path to destination $t_{i_{j}+1}$, for $j = 1, \hdots ,{r -1}$, 
   and $s_{i_{r}}$ has a path to $t_{i_{1}}$  in $\mathcal{G}-\mathcal{F}$.
   \end{claim}
The proof of Claim~\ref{not-gns-implies-sequence} is deferred to the end of this section. 
It follows from Claim~\ref{not-gns-implies-sequence}
that if $\widetilde{\mathcal{F}}$ is not a GNS cut-set of $\widetilde{\mathcal{G}}$, 
then ${\mathcal{G}^\prime - \mathcal{F}^\prime}$ contains a cycle,
contradicting the fact that $\mathcal{F}^\prime$ is a FES of $\mathcal{G}^\prime$.
We conclude that $\widetilde{\mathcal{F}}$ is a GNS cut-set in~$\widetilde{\mathcal{G}}$.
Note that $|\widetilde{\mathcal{F}}| = |\mathcal{F}^{\prime}|$.
Finally, the above implies that 
   \begin{align}
      \gnscut{\widetilde{\mathcal{G}}} 
      \quad \le \quad
      \min_{\mathcal{F}^\prime \text{ is a FES in } \mathcal{G}^\prime}|\mathcal{F}^\prime|.
      \label{eq:gnscut-Gtilde-le-min-FES-Gprime}
   \end{align}

Conversely, consider an arbitrary GNS cut-set $\widetilde{\mathcal{F}}$ in~$\widetilde{\mathcal{G}}$.
Let $\mathcal{F}^\prime$ be the (inverse) image of $\widetilde{\mathcal{F}}$ according to the mapping ${\cal M}$. 
We will show that $\mathcal{F}^\prime$ is an FES in~$\mathcal{G}^\prime$.

$\widetilde{\mathcal{F}}$ is a GNS cut-set.
Hence, there exists a permutation ${\pi:[k]\rightarrow [k]}$ such that
if~${\pi(i) \geq \pi(j)}$, no path exists from $\widetilde{s}_i$ to $t_j$ in ${\widetilde{\mathcal{G}}-\widetilde{\mathcal{F}}}$, $\forall i, j \in [k]$.
Assume, for the sake of contradiction, that $\mathcal{F}^\prime$ is not a FES in~$\mathcal{G}^\prime$, \textit{i.e.}, $\mathcal{G}^\prime - \mathcal{F}^\prime$ contains a cycle.
Any cycle in $\mathcal{G}^\prime$ has to include a link from the destination node $t_i$ to the source node $s_i$, for some ${i \in \lbrace 1, \hdots, k\rbrace}$,
\textit{i.e.}, it is of the form ${t_i, s_i, \hdots, t_i}$ including one or several source nodes.  
Using an argument identical to that in the proof of Theorem~\ref{thm:m-minus-mais-lt-gnscut}:
either
\textit{(i)} the cycle contains a path from $s_i$ to $t_i$, which is a path in $\mathcal{G}$, or 
\textit{(ii)} $\exists j: \pi(i) >\pi(j)$ and the cycle contains a path from $s_i$ to $t_j$.
This in turn implies (under the mapping $\mathcal{M}$)
that $\widetilde{\mathcal{G}}-\widetilde{\mathcal{F}}$ 
contains either a path from $\widetilde{s}_i$ to $t_i$, or a path from $\widetilde{s}_i$ to $t_j$ contradicting that $\widetilde{{\cal F}}$ is a GNS cut-set in $\widetilde{\mathcal{G}}$.
   We conclude that $\mathcal{F}^\prime$ is a FES of~$\mathcal{G}^\prime$,
   while by construction $|\mathcal{F}^\prime| = |\widetilde{\mathcal{F}}|$.
Finally, the above imply that 
   \begin{align}
      \min_{\mathcal{F}^\prime \text{ is a FES in } \mathcal{G}^\prime}|\mathcal{F}^\prime|
      \quad \le \quad
      \gnscut{\widetilde{\mathcal{G}}}.
      \label{eq:gnscut-Gtilde-ge-min-FES-Gprime}
   \end{align}
   The theorem follows from \eqref{eq:m-mais-eq-min-FES-Gprime}, \eqref{eq:gnscut-Gtilde-le-min-FES-Gprime} and \eqref{eq:gnscut-Gtilde-ge-min-FES-Gprime}.
   
   \subsubsection{Proof of Claim~\ref{not-gns-implies-sequence}}
   We prove the contrapositive statement;
   if $\mathcal{G}-\mathcal{F}$ contains
   no source-destination pair $s_{i}, t_{i}$ such that $s_{i}$ has a path to $t_{i}$,
   nor a sequence of $r \ge 2$ distinct indices $i_{1}, \hdots, i_{r} \in \lbrace 1, \hdots, k \rbrace$ 
   with the properties described in the claim, then $\mathcal{F}$ is a GNS cut. 
   
   Consider a directed graph $H$ on $k$ vertices labeled $1, \hdots, k$,
   with vertex $i$ corresponding to the $i$th source-terminal pair $s_{i}, t_{i}$ of~$\mathcal{G}-\mathcal{F}$. 
   A directed edge $(i,j)$ from vertex $i$ to vertex $j$ exists in $H$ if and only if a path from $s_{i}$ to $t_{j}$ exists in $\mathcal{G}-\mathcal{F}$.
   
   By assumption, $\mathcal{G}-\mathcal{F}$ contains no source-destination pair $s_{i}, t_{i}$ such that $s_{i}$ has a path to $t_{i}$. Hence, $H$ contains no self-loops.
   Further, it is straightforward to verify that $H$ contains a cyclic path $i_{1}, \hdots, i_{r}, i_{1}$, $r \ge 2$ if and only if the sequence of indices $i_{1}, \hdots, i_{r} \in \lbrace 1, \hdots, k \rbrace$ satisfies the property described in the claim.
   By assumption, no such sequence exists in $\mathcal{G}-\mathcal{F}$ either.
   Hence, $H$ is acyclic and has a topological ordering, 
   \textit{i.e.},
    a permutation ${\pi:[k]\rightarrow [k]}$ of the $k$ vertices such that for all $i,j \in [k]$, if $\pi(i) \ge \pi(j)$, then no edge exists from $i$ to $j$ . 
   In turn, if $\pi(i) \ge \pi(j)$, no path exists from source $s_i$ to destination $t_j$ in $\mathcal{G}-\mathcal{F}$.
   The existence of such a permutation implies that $\mathcal{F}$  is a GNS cut of~$\mathcal{G}$ (see Def.~\ref{definition-gnscut}).
$\qed$

\subsection{Proof of Theorem \ref{thm:app}}
\label{sec:proof:app}
Consider a multiple-unicasts network $\mathcal{G}$ with $k$ sources and $\Ecard$ (unit-capacity) edges in the set ${\cal E}$. 
Recall that~$\mathcal{G}^\prime$ is the network formed by setting the destination~$t_i$ as the tail of all source links of~$s_i$ in $\mathcal{G}$, ${\forall i \in \lbrace 1,\hdots, k\rbrace}$. 
We define~$\mathcal{L}$ to be the set of \textit{capacitated links} of $\mathcal{G}^{\prime}$:
$(a,b) \in \mathcal{L}$ if and only if there exists a link from $a$ to $b$ in $\mathcal{G}^{\prime}$ and the capacity of $(a,b)$, denoted by $c_{a,b}$, is equal to the number of unit links from $a$ to $b$ in $\mathcal{G}^{\prime}$.

First, we observe that to find the smallest feedback edge set in $\mathcal{G}^\prime$, it suffices to consider the set of capacitated links $\mathcal{L}$.
Consider a set $\mathcal{F}_{\mathsmaller{E}}^\prime \subseteq \mathcal{E}$ be a minimal feedback edge set of $\mathcal{G}^\prime$,
\textit{i.e.}, a minimal subset of \emph{unit-capacity} edges whose removal from~$\mathcal{G}^\prime$ yields an acyclic network.
If there exist multiple links from node $u$ to node $v$ in $\mathcal{G}^\prime$, then either all or none of them is included in $\mathcal{F}_{\mathsmaller{E}}^\prime$.
To verify that,
let $e$, $\widehat{e}$ be unit links from $u$ to $v$,
and $e \in \mathcal{F}_{\mathsmaller{E}}^\prime$, while $\widehat{e} \notin \mathcal{F}_{\mathsmaller{E}}^\prime$.
By construction, $\mathcal{G}^\prime - \mathcal{F}_{\mathsmaller{E}}^\prime$ is an acyclic network.
The minimality of $\mathcal{F}_{\mathsmaller{E}}^\prime$ implies that $\mathcal{G}^\prime$ contains a cycle $u^\prime, \hdots, {u \stackrel{e}{\rightarrow} v}, \hdots, u^\prime$ whose only edge contained in $\mathcal{F}_{\mathsmaller{E}}^\prime$ is $e$.
But, $\widehat{e} \notin \mathcal{F}_{\mathsmaller{E}}^\prime$ and hence $u^\prime, \hdots, u \stackrel{\widehat{e}}{\rightarrow} v, \hdots, u^\prime$ forms a cycle in $\mathcal{G}^\prime - \mathcal{F}_{\mathsmaller{E}}^\prime$, contradicting the fact that $\mathcal{F}_{\mathsmaller{E}}^\prime$ is a feedback edge set.

The second key observation is that since $\mathcal{G}$ is acyclic,
every cycle in $\mathcal{G}^\prime$ must include an edge from a destination node to a source node.
Equivalently, all cycles in $\mathcal{G}^\prime$ 
go through the set of nodes $\mathcal{S} = \lbrace s_1, \hdots, s_k \rbrace$,
\textit{i.e.}, the set of $k$ source nodes. 
The \emph{minimum weight subset-feedback edge set problem} (see \cite{even1998}) is the problem of finding a set of minimum weight edges that cuts all cycles passing through a \emph{specific} set of nodes. 
Finding a feedback edge set in~$\mathcal{G}^\prime$ is equivalent to solving the minimum weight subset-feedback edge set problem for the set of source nodes $\mathcal{S}$ (using the capacitated links $\mathcal{L}$ with weights coinciding with the corresponding capacity). 
The modified sphere growing approximation algorithm of \cite{even1998} with input $\mathcal{G}^\prime$ and $\mathcal{S}$ constructs a feedback edge set of weight within a $O({\log^2{|\mathcal{S}|}})$ factor from that of the minimum \emph{fractional} weighted feedback edge set of~$\mathcal{G}^\prime$, in time polynomial in~$|\mathcal{L}|$.
The weight of the minimum \emph{fractional} weighted feedback edge set coincides with the \emph{fractional cycle packing number} of ~$G$, $\rcp(G)$, where $G$ is the (reversed) line graph of~$\mathcal{G}^\prime$.
In other words, the aforementioned algorithm yields a feedback edge set of~$G$ with weight at most~$\rcp(G) \cdot O({\log^2{k}})$.


It is known \cite{compindex} that $\Ecard - \rcp(G)$ equals the broadcast rate of a vector-linear index code~$C$ for index coding instance with side-information graph~$G$.
In other words, there exists an index code~$C$, which can be defined over any field $\mathbb{F}$, that achieves broadcast rate ${\betavl^{\mathbb{F}}(\mathcal{C}; G) = \Ecard - \rcp(G)}$. 
The basis for the coding scheme is that every cycle in~$G$ saves a transmission.
This implies $\Ecard - \betavl(G)  \;\ge\;  \rcp(G)$. Further, by duality in Theorem \ref{thm:duality}, there is a feasible code for the relaxed-correlated sources problem on {\cal G} whose joint entropy rate is $\rcp(G)$.

\subsection{Proof of Theorem~\ref{thm:sep}}
\label{sec:proof:sep}
Our proof relies on the examples previously used to separate scalar linear and non-linear broadcast rates for index coding. 
We first show that the same examples imply separation results between vector-linear broadcast rates over two different fields. 
These results require the application of an \emph{uncertainty principle} (Lemma \ref{lem:bvl-uncertainty}) for vector-linear broadcast rates for index coding.\footnote{%
We note that Lemma \ref{lem:bvl-uncertainty}, although implied by results in~\cite{blasiak2013graph},
is to the best of our knowledge not explicitly stated in the literature.}
At first sight, those examples imply separation for only symmetric rates. 
However, we exploit certain symmetry properties that imply separation for sum-rates (asymmetric rates) when written as network coding instances.
We show that that there exist multiple-unicasts network instances for which optimal sum-rate (for independent sources) over two different prime fields~$\mathbb{F}_{p}$ and~$\mathbb{F}_{q}$ can be separated by a factor of $k^{1-o(1)}$.
This, along with Theorem~\ref{thm:app} concludes the proof.

\textit{Index coding -- optimum vector-linear broadcast rates over two different prime fields:}
The maximum vector-linear broadcast rate $\betavl^{\mathbb{F}}(G)$ over a field $\mathbb{F}$ for an index coding instance with side-information graph $G$ may vary significantly with the choice of $\mathbb{F}$.
We show that for any two prime fields $\mathbb{F}_p$ and $\mathbb{F}_q$,
there exist index coding instances $G$ for which 
${\betavl^{\mathbb{F}_p}(G)}$ and ${\betavl^{\mathbb{F}_q}(G)}$ for which 
${\betavl^{\mathbb{F}_q}(G) / \betavl^{\mathbb{F}_p}(G) \ge k^{1-o(1)}}$.

For any two prime fields $\mathbb{F}_p$ and $\mathbb{F}_q$,
Lubetzky and Stav \cite{lubetzky2009nonlinear} construct an undirected graph $G_u$
such that 
${\minrank_{\mathbb{F}_{q}}\left(G_u\right)\le k^{o(1)}}$
and 
${\minrank_{\mathbb{F}_{p}}\left(\overline{G}_{u}\right)\leq k^{o(1)}}$,
where $\overline{G}_{u}$ denotes the undirected graph complement.
The set of vertices of $G_{u}$ is the collection of all $s$-subsets of $[r]$.
An undirected edge between two vertices $X$ and $Y$ exists if and only if $ X \neq Y$ and $\lvert X \bigcap Y \rvert \equiv -1 (\mod p^b)$.
The parameters $r$, $s$ and $b$ depend on the choice of~$p$ and~$q$ (see Section~$2$ in~\cite{lubetzky2009nonlinear}).
Recall that~$\minrank_{\mathbb{F}}(G_u)$ equals the optimal broadcast rate on the index coding instance with side-information graph~$G_u$ among \emph{scalar} linear codes over~$\mathbb{F}$.
We conclude that 
\begin{align}
   \betavl^{\mathbb{F}_{p}} (G_u) \leq k^{o(1)}
   \quad \text{and} \quad
   \betavl^{\mathbb{F}_{q}} (\overline{G}_{u}) \leq k^{o(1)}. 
   \label{eq:2-fields-index-coding-sep}
\end{align}

By Lemma~\ref{lem:bvl-uncertainty} (\emph{uncertainty principle}), we know that for any field $\mathbb{F}_p$, $\betavl^{\mathbb{F}_{p}}\left(\overline{G}_{u}\right) \cdot \betavl^{\mathbb{F}_{p}}\left(G_u \right) \ge k$. This implies, 
\begin{align}
   \betavl^{\mathbb{F}_{q}} (\overline{G}_u) \leq k^{o(1)}
   \quad \text{and} \quad
   \betavl^{\mathbb{F}_{p}} (\overline{G}_{u}) \geq k^{1-o(1)}. 
   \label{eq:3-fields-index-coding-sep}
\end{align}

\textit{The network coding instance:}
Given $\overline{G}_{u}$,
we consider the multiple-unicasts network $\overline{\mathcal{G}}_{u}$ with $k$ sources constructed as follows\footnote{This is a slight modification of the construction in (Fig.~$2$)\cite{el2010index}.}:
\begin{enumerate}
   \item Create a source-destination pair $s_i$, $t_i$  for each vertex $i$ in $\overline{G}_{u}$.
   A unit-capacity link from  $s_i$ to $t_j$ exists in the network if and only if a (directed) edge $(i,j)$ exists in $\overline{G}_{u}$. 
   (Undirected edges in $\overline{G}_{u}$ are equivalent to two directed ones.)
   \item Introduce two nodes $a$ and $b$ with a unit-capacity link from $a$ to $b$.
   Connect $s_i$ to $a$ and $b$ to $t_i$ with unit-capacity links, $\forall i \in \lbrace 1, \hdots, k\rbrace$.
   \item Introduce a single unit-capacity source link (with no tail) for each $s_i$.
   Note that $\mincut{s_i,t_i}=1$ since there is a single path from $s_i$ to $t_i$; through link $(a,b)$.
\end{enumerate}
Any index code achieving broadcast rate $\beta$ in the index coding instance with side-information graph $\overline{G}_{u}$,
corresponds to a symmetric rate point $\left(\sfrac{1}{\beta}, \sfrac{1}{\beta}, \hdots \right)$ on the multiple-unicasts network instance $\overline{\mathcal{G}}_{u}$ with independent sources, and vice versa~\cite{el2010index}.

In principle, the optimum sum-rate on the multiple-unicasts network $\overline{\mathcal{G}}_{u}$ could be achievable by an asymmetric rate point.
We show that for~$\overline{\mathcal{G}}_{u}$ this is not the case; the optimal sum-rate can be achieved by a symmetric rate point.
Assume that an asymmetric rate point $\left(r_{1},\hdots,r_{k}\right)$ is achievable on $\overline{\mathcal{G}}_{u}$.
The original graph~$\overline{G}_{u}$ is vertex transitive:
given any two vertices~$v_1$ and~$v_2$ of~$\overline{G}_{u}$,
there is some automorphism ${f:V(G) \rightarrow V(G)}$ such that ${f(v_1) = v_2}$.
The multiple-unicasts network corresponding to the image of~$\overline{G}_{u}$ under~$f(\cdot)$ is identical to~$\overline{\mathcal{G}}_{u}$ (with rearranged sources),
implying that the rate point ${\bigl(r_{f^{-1}(1)},\hdots, r_{f^{-1}(k)}\bigr)}$ is also achievable on $\overline{\mathcal{G}}_{u}$.
$\overline{G}_{u}$ has a special structure: 
Each vertex is an $s$-subset of $[r]$ and a vertex for each subset exists. Further, the edge between two vertices is determined by only the size of the intersection of subsets corresponding to the pair of vertices. Therefore, the number of automorphisms~$f$ mapping~$v$ to~$u$ in~$\overline{G}_{u}$, 
is equal to the number of automorphisms mapping $v^\prime$ to $u$ in $\overline{G}_{u}$, and so on. This is because, all subsets in some sense are equivalent and therefore can be mapped in similar ways to a given subset.
Consider the entire collection of achievable rate tuples corresponding to all the automorphisms.
Then, the number of rate tuples in which user $u$ gets rate $r_v$ is equal to the number of tuples in which he gets rate $r_v^{\prime}$ and so on. The same holds for all users.
Therefore, time sharing among those rate tuples yields a symmetric rate point achieving sum-rate equal to $r_{1}+\hdots+r_{k}$.

In summary, any index code achieving broadcast rate~$\beta$ in the index coding instance with side-information graph~$\overline{G}_{u}$, corresponds to a symmetric rate point in the multiple-unicasts network coding instance on~$\overline{\mathcal{G}}_{u}$ with independent sources.
Further, the optimum sum-rate on~$\overline{\mathcal{G}}_{u}$ is achievable by a symmetric rate point.

We conclude that there exists a symmetric rate tuple for~$\overline{\mathcal{G}}_{u}$ that achieves sum-rate 
$\rateind\mathopen{}\left(\overline{\mathcal{G}}_{u},\mathbb{F}_{p}\right) = k/\betavl^{\mathbb{F}_{p}}\left(\overline{G}_{u}\right)$
and a tuple achieving 
$\rateind\mathopen{}\left({\overline{\mathcal{G}}_{u}},\mathbb{F}_{q}\right) =k/ \betavl^{\mathbb{F}_{q}}\left(\overline{G}_{u}\right)$.
In conjunction with \eqref{eq:3-fields-index-coding-sep}, we conclude that
\begin{align}
 \rateind\mathopen{}\left(\overline{\mathcal{G}}_{u},\mathbb{F}_{p}\right)
 \le k^{o(1)}
 \text{ and }
 \rateind\mathopen{}\left({\overline{\mathcal{G}}_{u}},\mathbb{F}_{q}\right) 
 \ge k^{1 -o(1)},
 \label{eq:2-fields-multicast-sep}
\end{align}
which proves the second part of the theorem.

\textit{Contradiction:}
Assume for the sake of contradiction that ${\ratecor\mathopen{}\left(\overline{\mathcal{G}}_{u},\mathbb{F}_p\right) < k^{1- \delta} \cdot  \rateind\mathopen{}\left({\overline{\mathcal{G}}_{u}},\mathbb{F}_p\right)} $.
By Theorem~\ref{thm:app}, we know that for any field $\mathbb{F}$,
$
   \rcp(\overline{G}_{u})
   \le
   \ratecor\mathopen{}\left(\overline{\mathcal{G}}_{u},\mathbb{F}\right)
   \le 
   \rcp(\overline{G}_{u})\cdot O(\log^2{k}).
$
Then, according to our assumption,
\begin{align}
   \rcp(\overline{G}_{u})
   \;<\;
   k^{1- \delta} \cdot \rateind\mathopen{}\left({\overline{\mathcal{G}}_{u}},\mathbb{F}_p\right).
   \label{eq:rcp-with-k-Rmu}
\end{align}
Also, $\rateind\mathopen{}\left({\overline{\mathcal{G}}_{u}},\mathbb{F}_q\right) \le \ratecor\mathopen{}\left({\overline{\mathcal{G}}_{u}},\mathbb{F}_q\right)$ which implies that
\begin{align}
   \rateind\mathopen{}\left({\overline{\mathcal{G}}_{u}},\mathbb{F}_q\right)
   \;\le\;
   \rcp(\overline{G}_{u}) \cdot O(\log^2{k}).
   \label{eq:rcp-with-Rmu-log2k}
\end{align}
Combinining \eqref{eq:rcp-with-k-Rmu} and \eqref{eq:rcp-with-Rmu-log2k}, we obtain
\begin{align}
   \rateind\mathopen{}\left({\overline{\mathcal{G}}_{u}},\mathbb{F}_q\right)
&  \;<\;
   k^{1- \delta} \cdot O(\log^2{k}) \cdot \rateind\mathopen{}\left({\overline{\mathcal{G}}_{u}},\mathbb{F}_p\right) \nonumber \\
&\;<\;
   k^{1- o(1)}\cdot \rateind\mathopen{}\left({\overline{\mathcal{G}}_{u}},\mathbb{F}_p\right), \nonumber
\end{align}
which contradicts \eqref{eq:2-fields-multicast-sep}.

\textit{Conclusion: }
For any prime field $\mathbb{F}_{p}$, we can construct a network ${\overline{\mathcal{G}}_{u}}$ such that
${\ratecor\mathopen{}\left(\overline{\mathcal{G}}_{u},\mathbb{F}_p\right) 
\ge 
k^{1- \delta} \cdot \rateind\mathopen{}\left({\overline{\mathcal{G}}_{u}},\mathbb{F}_p\right)}$.
This proves the first part of the theorem.
   
\subsection{Lower bounds on vector-linear Index Coding using strong products of graphs}\label{sec:strongprod} 
\begin{definition}
   \label{def:strong_product_of_graphs}
   (\textbf{Strong product of digraphs})
   The strong product $G \otimes H$ of two digraphs $G$ and $H$ is a digraph 
   on the set of vertices
   $V(G \otimes H) \triangleq \left\{ (u,v) : u \in V(G), v \in V(H) \right\} \nonumber$,
   with set of edges $E(G \otimes H)$ 
   that contains an edge from 
    $(u,v)$ to $(u^\prime, v^\prime)$ if and only if the following conditions both hold:
    \begin{enumerate}
     \item ${u=u^\prime} \,\vee\, (u,u^\prime) \in E(G)$
     \item ${v=v^\prime} \,\vee\, (v,v^\prime) \in E(H)$.
    \end{enumerate}
\end{definition}
\begin{definition}
	\label{def:digraph-complement}
	(\textbf{Complement of a digraph})
	The complement of a digraph ${G = (V, E)}$ is a digraph 
	${\overline{G}=(V, \overline{E})}$ on $V$, 
	where $\overline{E} = \left\lbrace (u, v) : u,v \in V, (u,v) \notin E \right\rbrace$.
\end{definition}

$\alpha(G)$ denotes an independent set of $G$, i.e. a set of vertices any two of which have no edge in either direction.
\begin{definition}
   \label{def:digraph-k-blowup}
   (\textbf{$k$-blowup of a digraph})
   The $k$-blowup of a digraph $G$ on $n$ vertices,
   denoted by $\blowup{G}{k}$,
   is a digraph on $k\cdot n$ vertices such that:
   \begin{itemize}
      \item $V(\blowup{G}{k})$ contains $k$ distinct vertices $\subpacket{v}{1}, \hdots, \subpacket{v}{k}$ for each vertex $v \in V(G)$. 
      \item An edge from $\subpacket{u}{i}$ to $\subpacket{v}{j}$, 
      $i, j \in [k]$ exists in $E(\blowup{G}{k})$
      if and only if ${(u, v) \in E(G)}$.
   \end{itemize}
\end{definition}
The definition of the $k$-blowup implies a surjective mapping from~$V(\blowup{G}{k})$ to~$V(G)$.
Throughout this paper, we use the convention that vertices $\subpacket{v}{i}$, $i=1,\hdots, k$, in $\blowup{G}{k}$ originate from (map to) vertex $v$ in $G$.
Conversely,~$v$ yields vertices $\subpacket{v}{1}, \hdots, \subpacket{v}{k}$ in $\blowup{G}{k}$.
Note that $\lbrace\subpacket{v}{i}\rbrace_{i=1}^{k}$ forms an independent set of cardinality~$k$ in~$\blowup{G}{k}$.
Further, each edge ${(u, v) \in E(G)}$ corresponds to a directed biclique\footnote{A directed biclique is a biclique with partition $(U, V)$ whose edges have 
been oriented from $U$ to $V$.}
$(\lbrace \subpacket{u}{i} \rbrace_{i=1}^{k}, \lbrace \subpacket{v}{i} \rbrace_{i=1}^{k})$ in $\blowup{G}{k}$.
\begin{definition}
   $\mais{G}$ of a directed graph $G$ is the cardinality of the largest set~${\widehat{V} \subseteq V(G)}$,
   such that the subgraph of~$G$ induced by~$\widehat{V}$ is acyclic.
\end{definition}
\begin{remark}
   For any undirected graph $G$, ${\mais{G} = \alpha(G)}$.
   An undirected graph $G$ can be considered as a digraph such that ${(u, v) \in E(G)} \Rightarrow {(v, u) \in E(G)}$.
   Hence, any edge $(u,v)$ in $G$ corresponds to a cycle ${u\rightarrow v\rightarrow u}$.
   In undirected graphs, 
   acyclic subgaphs coincide with independent sets. 
\end{remark}

Let $G$ be a digraph on $n$ vertices without self-loops.
We say that an~$n \times n$ matrix $\mathbf{A}$ over a finite field~$\mathbb{F}_q$ \textit{fits} $G$ if
\begin{inparaenum}[\itshape{(i)}]
 \item $A_{ii} \in \mathbb{F}_q \backslash \lbrace 0 \rbrace$, $\forall i \in [n]$, and
 \item ${A_{ij} = 0}$ ${\forall (i, j) \notin E(G)}$.
\end{inparaenum}
Let $\rank_q(\mathbf{A})$ denote the rank of $\mathbf{A}$ over $\mathbb{F}_q$.
\begin{definition}
	\label{def:minrk}
	$\minrank_q(G) = \min\mathopen{}\left\lbrace \rank_q(\mathbf{A})\,:\, \mathbf{A} \text{ fits } G \right\rbrace$.
\end{definition}

\begin{lemma} 
	\label{lemma:subgraph}
	Consider $m$ directed graphs $G_i$, $i = 1, \hdots, m$, and $m$ positive integers $k_i \in \mathbb{N}_{+}$.
	Let $G_{\alpha}$ be the strong product of~$\blowup{G_i}{k_i}$, $i=1, \hdots, m$, \textit{i.e.},
	\begin{align} 
		G_{\alpha} = \otimes_{i=1}^{m}\;\blowup{G_i}{k_i}, \nonumber
	\end{align}
	and $G_{\beta}$ be the $\left(\prod_{i=1}^{m} k_i\right)$-blowup of the strong product of $G_i$'s:
	\begin{align} 
		G_{\beta} = \blowup{\left(\otimes_{i=1}^{m}G_i\right)\mathclose{}}{
		\textstyle\prod_{i=1}^{m} k_i} 	\nonumber.
	\end{align}
	Then, $G_{\alpha}$ is (isomorphic to) a spanning subgraph of $G_{\beta}$.
\end{lemma}
\begin{proof}
	We establish a bijection ${f: V(G_{\alpha}) \rightarrow V(G_{\beta})}$ between the vertices of $G_{\alpha}$ and $G_{\beta}$,
	and show that two vertices $u,v \in V(G_{\alpha})$ are connected with an edge
	only if vertices $f(u),f(v)$ in $G_{\beta}$ are also connected.
	
	Let $n_i$ denote the number of vertices in $G_i$,
	and ${N \triangleq \prod_{i=1}^{m}n_i}$.
	By construction, $\blowup{G_i}{k_i}$
	is a graph on~${k_i \cdot n_i}$ vertices:
	$k_i$ vertices $\subpacket{v_i}{j}$, $j \in \lbrace 1, \hdots, k_i \rbrace$,
	for each ${v_i \in V(G_i)}$.
	The strong product~$G_{\alpha} = \otimes_{i=1}^{m}\blowup{G_i}{k_i}$
	is a graph on ${K \cdot N}$ vertices, where ${K \triangleq\prod_{i=1}^{m} \cdot k_i}$.
	The vertex set is formed as the cartesian product $V(G_1[k_1]) \times \cdots \times V(G_m[k_m])$: 
	the set of all $m$-tuples $(\subpacket{v_1}{j_1}, \hdots, \subpacket{v_m}{j_m})$ whose $i$th entry is a vertex 
	of $\blowup{G_i}{k_i}$.
	
	Similarly, the strong product $\otimes_{i=1}^{m}G_i$
	is a graph on a set of $N$ vertices.
	Each vertex corresponds to an $m$-tuple
	$\left( v_1, \hdots, v_m\right)$, where $v_i \in V(G_i)$.
	Its $K$-blowup is a graph $G_{\beta}$ on $K\cdot N$ vertices:
	$K$ vertices
	$\subpacket{(v_1, \hdots, v_m )}{j_1, \hdots, j_m}$, $j_i \in \lbrace 1, \hdots, k_i \rbrace$ for each vertex $(v_1, \hdots, v_m) \in V(\otimes_{i=1}^{m}G_i)$.
	Observe that we use an $m$-tuple $(j_1, \hdots, j_m)$ to enumerate the $K$ vertices in $G_{\beta}$
	originating from a single vertex of the strong product $\otimes_{i=1}^{m}G_i$.
	This alternative labeling establishes the one-to-one mapping between the vertices of $G_{\alpha}$ and $G_{\beta}$:
	\begin{align}
	 f\mathopen{}\bigl(\;
	      (
		 \subpacket{v_1}{j_1}, 
		 \hdots, 
		 \subpacket{v_m}{j_m}
	      )
	      \;\bigr)\mathclose{}
	      \;=\;
	      \subpacket{\left(v_1, \hdots, v_m\right)}{j_1, \hdots, j_m}.
	      \nonumber
	\end{align}
	
	It remains to show that two vertices $u=(\subpacket{u_1}{j_1}, \hdots, \subpacket{u_m}{j_m})$
	and $v= (\subpacket{v_1}{l_1}, \hdots, \subpacket{v_m}{l_m})$ in $G_{\alpha}$ are connected with an edge
	only if $f(u),f(v)$ in $G_{\beta}$ are also connected.
	By construction,
	an edge between $u$ and $v$ exists if an only if $\forall\; i \in \lbrace{1, \hdots m}\rbrace$:
	\begin{align}
		 {\subpacket{u_i}{j_i} = \subpacket{v_i}{l_i}}
		 \quad \vee \quad
		 \bigl( \subpacket{v_i}{j_i}, \subpacket{v_i}{l_i} \bigr) \in E\mathopen{}\left( \blowup{G_i}{k_i} \right).
		 \label{eq:condition-edge-in-H}
	\end{align}
	Note that, 
	${\subpacket{u_i}{j_i} = \subpacket{v_i}{l_i}}  \Rightarrow  {u_i = v_i \in V(G_i)}$.
	Further, an edge between $\subpacket{u_i}{j_i}$ and $\subpacket{v_i}{l_i}$ exists in $\blowup{G_i}{k_i}$ 
	if and only if an edge between $u_i$ and $v_i$ exists in $G_i$.
	That is,
	\begin{align}
	   ( \subpacket{u_i}{j_i}, \subpacket{v_i}{l_i} ) \in E\left( \blowup{G_i}{k_i} \right)
	   \quad
	   \Leftrightarrow
	   \quad
	   \left( u_i ,v_i \right) \in E\mathopen{}\left(G_i\right).
	   \nonumber
	\end{align}
	Hence, the existence of an edge between $u$ and $v$ in $G_{\alpha}$, implies that ${\forall i \in \lbrace{1, \hdots m}\rbrace}$,
	\begin{align}
	   u_i = v_i
	   \;
	   \vee
	   \;
	   \left( u_i ,v_i \right) \in E(G_i),\quad
	   u_i, v_i \in V(G_i).
	   \nonumber
	\end{align}
	It follows that $\otimes_{i=1}^{m}G_i$ contains an edge between $(u_1, \hdots, u_m)$ and $(v_1, \hdots, v_m)$,
	and in turn $G_{\beta}$ contains an edge between
	${f(u) =\subpacket{(u_1,\hdots, u_m)}{j_1, \hdots, j_m}}$ and 
	${f(v) =\subpacket{(v_1, \hdots, v_m)}{l_1, \hdots, l_m}}$, which completes the proof.
\end{proof}
\begin{corollary}
	\label{corollary:indep-set-mais-inequalities}
	For the graphs $G_{\alpha}$ and $G_{\beta}$ of Lemma~\ref{lemma:subgraph},
	\begin{enumerate}
		\item \label{part-indep-set} $\alpha \left(G_{\alpha}\right)    \ge  \alpha \left( G_{\beta}\right)$, 
	   	\item \label{part-mais} $\mais{ G_{\alpha}}    \ge  \mais{ G_{\beta}}$.
	   \end{enumerate}
\end{corollary}
\begin{proof}
   The corollary is a straightforward consequence of the fact
   that $G_{\alpha}$ is a spanning subgraph of $G_{\beta}$:
   \paragraph*{Part (\textit{i})}
   Consider the maximum independent set~$\mathcal{I}(G_{\beta})$
   of~$G_{\beta}$.
   Since $G_{\alpha}$ is a spanning subgraph of $G_{\beta}$,
   no edge exists between the vertices in the inverse image $f^{-1}\left( \mathcal{I}(G_{\beta})\right)$ of $\mathcal{I}(G_{\beta})$.
   In other words, $f^{-1}\left( \mathcal{I}(G_{\beta})\right)$ is an independent set in~$G_{\alpha}$.
   Given that $f(\cdot)$ is a bijection, we conclude that $|\mathcal{I}(G_{\alpha})| \ge |f^{-1}\left( \mathcal{I}(G_{\beta})\right)| =\mathcal{I}(G_{\beta})$.
   \paragraph*{Part (\textit{ii})}
   Let $\mathcal{A}(G)$ denote the largest subset of vertices in $G$ such that the subgraph induced by $\mathcal{A}(G)$ is acyclic.
   Consider the vertices of $G_{\alpha}$ in the inverse image of $\mathcal{A}(G_{\beta})$, $f^{-1}\mathopen{}\left(\mathcal{A}(G_{\beta})\right)$.
   The corresponding subgraph of $G_{\alpha}$ contains no cycle.
   Given that $f(\cdot)$ is a bijection, we conclude that 
   $\mais{G_{\alpha}} 
    \ge |f^{-1}\mathopen{}\left( \mathcal{A}(G_{\beta})\right)| =\mathcal{A}(G_{\beta})$.
\end{proof}
\begin{lemma}
   \label{lemma:mais-of-strong-product}
   For any two directed graphs $G$ and $H$,
   \begin{align}
      \mais{G \otimes H}
      \ge
      \mais{G} \cdot \mais{H}. \nonumber
   \end{align}
\end{lemma}
\begin{proof}
Let $\mathcal{A}(G) \subseteq V(G)$ denote the largest set of vertices such that the subgraph of $G$ induced by $\mathcal{A}(G)$ is acyclic, \textit{i.e.}, $|\mathcal{A}(G)|= \mais{G}$.
Consider the set
\begin{align}
   \mathcal{S} \triangleq
   \left\lbrace (u, v) : u \in \mathcal{A}(G), v \in \mathcal{A}(H) \right\rbrace
   \;\subseteq \;V(G \otimes H).
   \nonumber
\end{align}
We show that the subgraph of $G \otimes H$ induced by $\mathcal{S}$ is acyclic.   
Assume for the sake of contradiction that there exists a subset $\mathcal{C} \subseteq \mathcal{S}$ of $m > 1$ vertices that lie on a directed cycle: $(u_0, v_0) \rightarrow \cdots \rightarrow (u_{m-1}, v_{m-1}) \rightarrow (u_0, v_0)$.
The existence of edges between consecutive vertices implies that the following is true for all $i \in \lbrace 0, \hdots, m-1 \rbrace$:
\begin{align}
   u_i = u_{i+1\bmod m} \;\vee\; (u_i,\;u_{i+1 \bmod m}) \in E(G). \nonumber
\end{align}
Hence, the sequence of vertices $u_0, \hdots, u_{m-1}, u_0$ in $\mathcal{A}(G)$ either consists of a single vertex repeating $m+1$ times, or forms a cyclic path in $G$, ignoring transitions where consecutive vertices in the sequence are the same vertex in $G$.
The same holds for the sequence of vertices $v_0, \hdots, v_{m-1}, v_0$ in $\mathcal{A}(H)$.
At least one of the two sequences must contain more than one distinct vertices; otherwise, $|\mathcal{C}|=1$.
Therefore, either the subgraph of $G$ induced by $\mathcal{A}(G)$ or the subgraph of $H$ induced by $\mathcal{A}(H)$ contains a cycle, contradicting the definition of $\mathcal{A}(\cdot)$.
We conclude that the subgraph of $G \otimes H$ induced by $\mathcal{S}$ is acyclic and $\mais{G \otimes H} \ge |\mathcal{S}| =  \mais{G} \cdot \mais{H}$.
\end{proof}
\begin{lemma}
	\label{alphalemma}
	For any directed graph $G$ and its $k$-blowup $\blowup{G}{k}$,
	\begin{enumerate}
		\item \label{alphalemma:part-indep-set} $ \alpha (\blowup{G}{k}) = k \cdot \alpha(G)$,
	   	\item \label{alphalemma:part-mais} $\mais{\blowup{G}{k}} = k \cdot \mais{G}$.
	\end{enumerate}
\end{lemma}
\begin{proof}
   By construction $\blowup{G}{k}$ contains $k$ distinct unconnected vertices ${\subpacket{v}{1}, \hdots, \subpacket{v}{k}}$ for each vertex $v \in V(G)$.
   \paragraph*{Part \textit{(i)}}
   Consider the set  
   \begin{align}
   \mathcal{S} \triangleq \bigl\lbrace v \in V(G) : \subpacket{v}{i} \in \ind{\blowup{G}{k}}, i \in [k] \bigr\rbrace.
   \end{align}
   $\mathcal{S}$ is an independent set in $G$:
   if an edge exists between two vertices $u, v \in \mathcal{S}$ in $G$,
   then $(\subpacket{v}{i},\; \subpacket{u}{j}) \in E(\blowup{G}{k}), \; \forall i, j \in [k]$, contradicting the fact that $\subpacket{v}{i}, \subpacket{u}{j} \in \ind{\blowup{G}{k}}$ for some $i, j \in [k]$.
   Taking into account that at most $k$ vertices in $\mathcal{I}(\blowup{G}{k})$ (and generally in $\blowup{G}{k}$) originate from a single vertex in $G$, we conclude that
   \begin{align}
      \alpha(G) \ge |\mathcal{S}| = \tfrac{1}{k} \cdot \alpha(\blowup{G}{k}).
      \label{eq:indG-gt-indGk-over-k}
   \end{align}     
   $ \mathcal{T}\triangleq
	    \bigl\lbrace
		  \subpacket{v}{i} \in V(\blowup{G}{k}), \; i = 1, \hdots, k: v \in \ind{G}
	    \bigr\rbrace $
   is an independent set in $\blowup{G}{k}$.
   By construction, no edge exists between $\subpacket{v}{i}$ and $\subpacket{v}{j}$, $i \neq j$.
   Further, an edge between two vertices $\subpacket{v}{i}, \subpacket{u}{j} \in \mathcal{T}$ in $\blowup{G}{k}$ exists only if an edge exists between $v$ and $u$ in $G$. But $v, u \in \ind{G}$.
   Thus,
   \begin{align}
      \alpha(\blowup{G}{k}) \ge |\mathcal{T}| = k \cdot \alpha(G).
      \label{eq:indGk-gt-indG-times-k}
   \end{align}
   The desired result follows from \eqref{eq:indG-gt-indGk-over-k} and \eqref{eq:indGk-gt-indG-times-k}.
   \paragraph*{Part \textit{(ii)}}
   Let $\mathcal{A}(G) \subseteq V(G)$ denote the largest set of vertices such that the subgraph of $G$ induced by $\mathcal{A}(G)$ is acyclic, \textit{i.e.}, $\mais{G} = |\mathcal{A}(G)|$.
   Consider the set
   \begin{align}
	    \mathcal{S} \triangleq
	    \bigl\lbrace
		  v \in V(G) : \subpacket{v}{i} \in \mathcal{A}(\blowup{G}{k}), i \in [k]
	    \bigr\rbrace.
	    \nonumber
   \end{align}
   The subgraph of $G$ induced by $\mathcal{S}$ contains no cycle.
   To verify that, observe that if there exists a subset
   $\mathcal{C} \subseteq \mathcal{S}$ of $m$ vertices forming a directed cycle
   $v_1 \rightarrow \cdots \rightarrow v_m \rightarrow v_1$,
   then $\subpacket{v_1}{i_1} \rightarrow \cdots \rightarrow \subpacket{v_m}{i_m} \rightarrow \subpacket{v_1}{i_1}$ forms a cycle in $\blowup{G}{k}$ for any combination of $i_1, \hdots, i_m$.
   Therefore, a cycle in $\blowup{G}{k}$ can be formed using vertices in $\mathcal{A}(\blowup{G}{k})$, contradicting the definition of $\mathcal{A}(\cdot)$.
   It follows that 
   \begin{align}
      \mais{G} \ge |\mathcal{S}| 
      = 
      \tfrac{1}{k}  \cdot \mais{ \blowup{G}{k}}.
      \label{eq:maisG-gt-maisGk-over-k}
   \end{align}
   Conversely, the set
   \begin{align}
	    \mathcal{T} \triangleq
	    \bigl\lbrace
		  \subpacket{v}{i} \in V(\blowup{G}{k}), \;  i=1, \hdots, k: v \in \mathcal{A}(G)
	    \bigr\rbrace
	    \nonumber
   \end{align}
   induces an acyclic subgraph in $\blowup{G}{k}$:
   if there exist a subset of vertices in $\mathcal{T}$ forming a cycle 
   $\subpacket{v}{i} \rightarrow \subpacket{u}{j} \rightarrow \cdots \rightarrow \subpacket{v}{i}$ in $\blowup{G}{k}$,
   then 
   ${v \rightarrow u \rightarrow \cdots \rightarrow v}$ is a cycle in $G$, 
   contradicting the fact that ${v, u, \hdots \in \mathcal{A}(G)}$.
   Thus,
   \begin{align}
      \mais{\blowup{G}{k}} \ge |\mathcal{T}| = k \cdot \mais{G}.
      \label{eq:maisGk-gt-maisG-times-k}
   \end{align}
      The desired result follows from~\eqref{eq:maisG-gt-maisGk-over-k} and~\eqref{eq:maisGk-gt-maisG-times-k}.
\end{proof}
\begin{lemma}
   \label{lemma:blowup-minrnk-gt-mais-product}
   For any $k$, $m \in \mathbb{N}_{+}$,
   \begin{align}
	 \tfrac{1}{k} \cdot \minrank_q\left(\blowup{G}{k}\right)
	 \ge
	 \mais{G^m}^{1/m}.
	 \nonumber
   \end{align}
\end{lemma}
\begin{proof}
	We have
	\begin{align}
		\minrank_q^m{\left(\blowup{G}{k}\right)}
		&\stackrel{(\alpha)}{\ge}
		\minrank_q{\left(\otimes^m  (\blowup{G}{k})\right)}
		\;\stackrel{(\beta)}{\ge}
		\mais{\otimes^m  (\blowup{G}{k})}\nonumber \\
		&\stackrel{(\gamma)}{\ge}
		\mais{\blowup{(\otimes^m G)}{k^m}}
		\;\stackrel{(\delta)}{\ge}
		k^m \cdot \mais{\otimes^m G}, \nonumber
	\end{align}
	where 
	$(\alpha)$ follows from the submultiplicativity of the $\textsf{minrk}_q(\cdot)$ function \cite{lubetzky2009nonlinear},
	$(\beta)$ from 
	\cite{bar2011index},
	$(\gamma)$ follows from Corollary~\ref{corollary:indep-set-mais-inequalities}, and 
	$(\delta)$  from Lemma~\ref{alphalemma}.
\end{proof}
\begin{lemma}
   \label{lem:bvl-uncertainty}
   $
	 \betavl^{\mathbb{F}}(G) \cdot
	 \betavl^{\mathbb{F}}(\overline{G})
	 \ge
	 n
	 \nonumber.
   $
\end{lemma}
\begin{proof}
	For any $k_1, k_2 \in \mathbb{N}_{+}$,
	 \begin{align}
	   \minrank_{\mathbb{F}}& \left( \blowup{G}{k_1} \right) \cdot
	   \minrank_{\mathbb{F}} \left( \blowup{\overline{G}}{k_2} \right) \nonumber \\
	   &\stackrel{(\alpha)}{\ge}
	   \minrank_{\mathbb{F}} \left( \blowup{G}{k_1} \otimes  \blowup{\overline{G}}{k_2}\right)
	   \stackrel{(\beta)}{\ge}
	   \alpha \left( \blowup{G}{k_1} \otimes  \blowup{\overline{G}}{k_2}\right) \nonumber \\
	   &\stackrel{(\gamma)}{\ge}
	   \alpha \left( \blowup{(G \otimes \overline{G})}{k_1 \cdot k_2}\right)
	   \stackrel{(\delta)}{\ge}
	   k_1 \cdot k_2 \cdot \alpha \left(G \otimes \overline{G}\right) \nonumber \\
	   &\stackrel{}{\ge}
	   k_1 \cdot k_2 \cdot n, \label{eq:lb-on-prod-minrk}
	 \end{align}
	where 
	$(\alpha)$ follows from the submultiplicativity of the $\textsf{minrk}_q$ function \cite{lubetzky2009nonlinear},
	$(\beta)$ from Lemma~\ref{lemma:blowup-minrnk-gt-mais-product} and the fact that $\mais{G} \ge \alpha(G)$, 
	$(\gamma)$ from Corollary~\ref{corollary:indep-set-mais-inequalities}, and 
	$(\delta)$  from Lemma~\ref{alphalemma}.
	The last inequality follows from the fact that $\alpha \left(G \otimes \overline{G}\right) \ge n$.
	To verify that, observe that the set $\lbrace (u, u): u \in V(G) \rbrace \subset V(G \times \overline{G})$ forms an independent set of cardinality $n$ in $G \times \overline{G}$.
	Inequality~\eqref{eq:lb-on-prod-minrk} holds for all $k_1$ and $k_2$.
	The desired result follows from Theorem~\ref{vectorlinear}.
\end{proof}
Extending the arguments for optimality of scalar $\minrank$ in \cite{bar2011index} for scalar linear index coding, it can be shown that the optimum broadcast rate of vector-linear index codes over all fields can be shown to be infimum of the normalized minranks of graph blowups.
\begin{theorem}[\cite{blasiak2013graph}]
	\label{vectorlinear}
	  $\betavl^{\mathbb{F}} \left( G \right) 
	  = 
	  \inf \limits_{k}\frac{\minrank_{\mathbb{F}} \left( G[k] \right)}{k}$.
\end{theorem} 
\begin{theorem} \footnote{%
\cite{blasiak2013graph} contains proofs for Theorems~\ref{Shannonbound} and~\ref{thm:shannoncapbound}(i).
Lemma \ref{lem:bvl-uncertainty} can be also inferred, but is not explicitly stated.
We give a different proof based on strong products of graphs that also generalizes bounds by tensorizing MAIS(G) instead of just tensorizing $\alpha(G)$.}
\label{Shannonbound}
$\Theta(G) \le \betavl(G)$.

\end{theorem}
\begin{proof}
      For any $k, m \in \mathbb{N}_{+}$,
      \begin{align}
         \alpha \left( \otimes^m \blowup{G}{k}  \right) 
        \;\stackrel{\mathrm{Cor.}~\ref{corollary:indep-set-mais-inequalities}}{\ge} \;
        \alpha \left( \blowup{(\otimes^m G)}{k^m} \right)
         \;\stackrel{\mathrm{Lem.}~\ref{alphalemma}}{=}\;
         k^m \alpha \left( \otimes^m G\right).
	 \nonumber
       \end{align}   
    Hence, by the definition of Shannon capacity,
     \begin{align}
	 \Theta \left(G \right)
	 &=
	    \lim_{m \rightarrow \infty} \alpha(\otimes^m G)^{1/m}  \le
	    \lim_{m \rightarrow \infty} \frac{1}{k}\alpha(\otimes^m \blowup{G}{k})^{1/m} \nonumber \\
	 &=
	    \frac{1}{k} \Theta \left( \blowup{G}{k} \right)
	 \le \frac{1}{k}\minrank_{\mathbb{F}} \left( \blowup{G}{k} \right), 
	 \qquad {\forall k \in \mathbb{N}_{+}}.
	 \nonumber
      \end{align}
      The desired result follows from Theorem~\ref{vectorlinear}.
\end{proof}
\begin{theorem}\label{thm:shannoncapbound}
   For any graph $G$ and $m \in \mathbb{N}_{+}$,
   \begin{enumerate}
   \item    $
		\alpha(\otimes^m G)^{1/m}
	 \le \Theta(G)
	 \le \betavl(G).
      $
   \item
      $
	 \mais{\otimes^m G}^{1/m}
	 \le \betavl(G).
      $
   \end{enumerate}
\end{theorem}
\begin{proof}
   For the second part, by Theorem~\ref{vectorlinear}, 
   and Lemma~\ref{lemma:blowup-minrnk-gt-mais-product},
   $\betavl^{\mathbb{F}} \ge \mais{\otimes^m G}^{1/m}$, for any field $\mathbb{F}$.
   The desired result follows taking infimum over all fields.
   The first part follows from the fact that $\alpha(\otimes^m G) \le \mais{\otimes^m G}$, in conjunction with the definition of $\Theta(G)$
   and Theorem~\ref{Shannonbound}.
\end{proof}

\section{Old Approximation}

\pagenumbering{arabic}

\end{document}